\documentclass[12pt,nodate]{article}
\usepackage{amssymb,amsmath,amsthm,color,tikz}
\usepackage{fullpage,comment}
\usepackage[margin=0.8in]{geometry}
\theoremstyle{plain}
\newtheorem{thm}{Theorem}
\newtheorem{lem}{Lemma}

\theoremstyle{definition}
\newtheorem{rem}{Remark}
\newtheorem{manualtheoreminner}{Theorem}
\newenvironment{manualtheorem}[1]{%
	\renewcommand\themanualtheoreminner{#1}%
\manualtheoreminner
}{\endmanualtheoreminner}

\usepackage[colorlinks=true,linkcolor=blue,urlcolor=blue,citecolor=blue]{hyperref}
\usepackage{doi}
\newcommand{\mc}{\mathcal}

\newcommand{\B}{{\mathcal B}}

\newcommand{\D}{\tilde{D}}

\newcommand{\tr}{{\rm Tr}}
\newcommand{\id}{{\mathbb I}}

\newcommand{\ket}[1]{| #1 \rangle}
\newcommand{\bra}[1]{\langle #1 |}
\newcommand{\braket}[2]{\left\langle #1\! \mid\! #2 \right\rangle}

\newcommand{\R}{{\mathcal R}}

\renewcommand{\H}{\mathcal{H}}
\renewcommand{\l}{\Lambda}
\renewcommand{\P}{\mathcal{P}}
\newcommand{\K}{\mathcal{K}}
\newcommand{\dd}{\mathrm{d}}
\newcommand{\norm}[1]{\left\lVert#1\right\rVert}

\newcommand*\circled[1]{\tikz[baseline=(char.base)]{
            \node[shape=circle,draw,inner sep=1.1pt] (char) {#1};}}

\title{Revisiting the equality conditions of the data processing inequality for the sandwiched R\'enyi divergence}
\author{Jinzhao Wang\footnote{jinzwang@phys.ethz.ch}, Henrik Wilming\footnote{henrikw@phys.ethz.ch}}
\date{\small Institute  for  Theoretical  Physics,  ETH  Zurich,  8093  Zurich,  Switzerland\\ }
\begin{document}
\maketitle
\abstract{We provide a transparent, simple and unified treatment of recent results on the equality conditions for the data processing inequality (DPI) of the sandwiched quantum R\'enyi divergence, including the statement that equality in the data processing implies recoverability via the Petz recovery map for the full range of $\alpha$ recently proven by Jen\v cov\'a. We also obtain a new set of equality conditions, generalizing a previous result by Leditzky et al.}

\section{Introduction}
One of the most fundamental relations in quantum information theory is the \emph{data-processing inequality} (DPI)
\begin{align}\label{eq:dpi1}
D(\l(\rho) || \l(\sigma)) \leq D(\rho||\sigma)
\end{align}
where $D(\cdot \| \cdot)$ is the quantum relative entropy and $\l$ is a completely positive trace-preserving (CPTP) map. Physically, it means that it's harder to distinguish two quantum states after a quantum channel acted on them. 
The DPI holds more generally in the context of von Neumann Algebras, but to avoid technicalities, we focus on finite dimensional quantum systems in this paper.
We also assume for simplicity that all density matrices are strictly positive.

There are various generalizations of the relative entropy to a family of quantum R\'enyi divergences, largely motivated by the corresponding classical notions. Some common ones studied in the literature are the Petz quantum R\'enyi divergence $\overline{D}_\alpha$\cite{petz1986quasi,hiai2011quantum} and the sandwiched quantum R\'enyi divergence $\D_\alpha$~\cite{muller2013quantum,wilde2014strong}. There are also other more general $f$-divergences and maximal $f$-divergences studied in the literature \cite{hiai2011quantum,hiai2017different}. The proofs of their DPI's and can all be unified with a single approach using the relative modular operator~\cite{wilde2018optimized}. This method is largely due to Araki~\cite{araki1976relative}, and later elaborated upon by Petz, Nielson, Wilde and many others~\cite{petz1986quasi,petz2003monotonicity,nielsen2005simple,hiai2011quantum,wilde2018optimized}. A particular advantage of this approach is that it more naturally generalizes to the setting of infinite-dimensional von~Neumann algebras required in quantum field theory \cite{witten2018aps}. The sandwiched R\'enyi divergence has also recently acquired some attention in quantum field theory \cite{lashkari2019constraining,faulkner2020approximate,moosa2020renyi,hollands2020variational}. 

Our analysis in this work is also based on the relative modular operator machinery, and we shall study the equality conditions for saturating the DPI for the sandwiched R\'enyi divergences~\cite{leditzky2017data,jenvcova2017preservation,jenvcova2018renyi,jenvcova2017r}. Our main aim in this work is to give a simple, transparent and unified treatment of the equality conditions known in the literature, which we discuss below. But we also find new equality conditions which generalize the known conditions. 

The sandwiched R\'enyi divergences were first introduced in~\cite{muller2013quantum,wilde2014strong} and can be defined as (for $\alpha\in[-1,0)\cup(0,1)$)
\begin{equation}
\D_\alpha(\rho||\sigma):= \frac{1-\alpha}{\alpha}\log\tr(\rho^\frac12\sigma^{-\alpha}\rho^\frac12)^\frac{1}{1-\alpha}.
\end{equation} 
In the literature it is more customary to use a different parametrization of the R\'enyi divergences, which in our convention amounts to replacing $\alpha$ with $n$ fulfilling $\alpha=\frac{n-1}{n}$, so that $n\in[1/2,1)\cup(1,\infty)$. In our convention, the relative entropy can be obtained as the $\alpha\rightarrow 0$ limit of $\D_\alpha$, and the DPI holds in the range $\alpha\in [-1,0)\cup(0,1)$. We restrict to this range and, since we consider strictly positive $\rho,\sigma$, the support of $\rho$ always lies in the support of $\sigma.$

Due to the fundamental importance of the data-processing inequality, significant effort has been devoted to understanding when it is saturated, both for the quantum relative entropy and other divergences. 
For example, it is proved in~\cite{leditzky2017data} that a necessary and sufficient equality condition for the sandwiched R\'enyi divergence is as follows.
\begin{thm}[Leditzky, Rouz\'e and Datta]\label{theorem:equality}
For $\alpha\in[-1,0)\cup(0,1)$, density matrices $\sigma,\rho$ and a CPTP map $\l$, we have $\D_\alpha(\sigma||\rho) = \D_\alpha(\Lambda(\sigma)||\Lambda(\rho))$ if and only if
\begin{equation}\label{eq:equalitycondition1}
\sigma^{-\frac{\alpha}{2}}(\sigma^{-\frac{\alpha}{2}}\rho\sigma^{-\frac{\alpha}{2}})^\frac{\alpha}{1-\alpha}\sigma^{-\frac{\alpha}{2}}= \l^*\left(\l(\sigma)^{-\frac{\alpha}{2}}[\l(\sigma)^{-\frac{\alpha}{2}}\l(\rho)\l(\sigma)^{-\frac{\alpha}{2}}]^\frac{\alpha}{1-\alpha}\l(\sigma)^{-\frac{\alpha}{2}}\right)\,.
\end{equation}
\end{thm}
The proof relies on a variational formula of $\D_\alpha$ introduced in~\cite{frank2013monotonicity}. Below, we will use a different variational formula, and it will turn out to provide an alternative, simple proof of the above equality conditions. The physical meaning of an algebraic equality  such as the one above is, however, obscure. The more operational 
perspective on saturating the DPI is related to the topic of recoverability.  It says that the saturation of the DPI is equivalent to the existence of a quantum channel $\mc R$ such that $\mc R\circ \l(\rho)=\rho$ and $\mc R\circ\l(\sigma)=\sigma$, which is also known as the recoverability (sufficiency) of the quantum channel with respect to $\{\rho,\sigma\}$. The canonical example is that saturating the DPI for the quantum relative entropy for some $\{\l,\rho,\sigma\}$ is equivalent to recoverability of the triple.

Indeed, Petz analyzed the case of equality for the relative entropy in~\cite{petz2003monotonicity}, and found the following equality condition for $\l$ being a partial trace $\tr_B$:
\begin{equation}\label{eq:relativeentequality}
\sigma_{AB}^{\beta}\rho_{AB}^{-\beta} = \sigma_A^{\beta}\rho_{A}^{-\beta},\quad\forall\beta\in\mathbb{C}.
\end{equation}
By choosing $\beta=-\frac12$ and taking the modulus square on both sides, one obtains
\begin{equation}
\sigma_{AB}^{-\frac12}\rho_{AB} \sigma_{AB}^{-\frac12} = \sigma_A^{-\frac12}\rho_{A}\sigma_A^{-\frac12},
\end{equation}
which can be rearranged as
\begin{equation}
\rho_{AB} = \sigma_{AB}^{\frac12}\sigma_A^{-\frac12}\rho_{A}\sigma_A^{-\frac12}\sigma_{AB}^{\frac12} := \R_{\sigma,\tr_B}(\rho_A),
\end{equation}
where $\R_{\sigma,\l}$ is the \emph{Petz recovery map}, which is generally defined as
\begin{equation}\label{eq:petzmap}
\R_{\sigma,\l}(\cdot) = \sigma^\frac12\l^*\left(\l(\sigma)^{-\frac12} \cdot \l(\sigma)^{-\frac12}\right)\sigma^\frac12.
\end{equation} 
Here, $\Lambda^*$ denotes the adjoint channel to $\Lambda$. By definition the Petz map satisfies $\R_{\sigma,\l}(\sigma)=\sigma,$ so this proves that the saturation of the DPI for the quantum relative entropy is equivalent to recoverability. In fact, (\ref{eq:relativeentequality}) holds for the more general class of quantum $f$-divergences as well, so we know that the recoverability can also be characterized by the equality in their DPI's.
  
Let us now come back to the sandwiched R\'enyi divergence. Unfortunately, the equality condition (\ref{eq:equalitycondition1}) itself does not seem to hint at the recoverability of the quantum channel.  
Nevertheless, the problem was recently resolved by Jen\v cov\'a in two separate works dealing with the range $\alpha \in (-1,0)$ and $\alpha \in (0,1)$ respectively \cite{jenvcova2017preservation,jenvcova2017r}, yielding the following theorem:
\begin{thm}[Jen\v cov\'a]\label{theorem:recoverability}
	For $\alpha\in(-1,0)\cup(0,1)$, density matrices $\sigma,\rho$ and a quantum channel $\l$, $\l$ is recoverable with respect to $\{\rho,\sigma\}$ if and only if $\D_\alpha(\rho||\sigma) = \D_\alpha(\l(\rho)||\l(\sigma))$.
\end{thm}
In fact Jen\v cov\'a's results are more general than the above theorem, in particular they also cover the setting of general, finite-dimensional von~Neumann algebras. 
To obtain these results, she first identified the recoverability for a different pair $\{\tau,\sigma\}$ where $\tau$ is some other state composed by $\rho,\sigma$. Then the recoverability for $\rho,\sigma$ is implied by a universal factorization structure for these recoverable states \cite{petz1988sufficiency,hayden2004structure,jenvcova2006sufficiency}. Her arguments make use of an interpolating family of $L_p$-norms in the general setting of non-commutative $L_p$-spaces, and different techniques are needed to tackle the two $\alpha$ ranges separately. In this work, we demonstrate that one can in fact directly obtain the Petz recovery map by analyzing the DPI equality conditions of the sandwiched R\'enyi divergence \`a la Petz \cite{petz2003monotonicity}, and it gives us an elementary recoverability proof for the entire range $\alpha\in (-1,0)\cup(0,1)$. 

Following this line of thought, we also find new algebraic conditions that characterize the equality in DPI, summarized in the following theorem.
\begin{thm}\label{theorem:newequality}
For $\alpha\in(-1,0)\cup(0,1)$, density matrices $\sigma,\rho$ and a CPTP map $\Lambda$, we have $\D_\alpha(\sigma||\rho) = \D_\alpha(\Lambda(\sigma)||\Lambda(\rho))$ if and only if for all $\beta\in\mathbb{C}$ and $\alpha\in(0,-1)\cup(0,1)$
\begin{equation}\label{eq:newresult}
\sigma^\beta(\rho\sigma^{-\alpha})^{\frac{\beta}{\alpha-1}}= \l^*\left(\l(\sigma)^\beta[\l(\rho)\l(\sigma)^{-\alpha}]^{\frac{\beta}{\alpha-1}}\right)\,.
\end{equation}
\end{thm}
Note that the matrix power $(\rho\sigma^{-\alpha})^{\frac{1}{\alpha-1}}$ is well-defined as the $\rho\sigma^{-\alpha}$ has positive eigenvalues. Our Theorem \ref{theorem:newequality} is a generalization of Theorem \ref{theorem:equality}, as the former implies the latter, c.f. Remark \ref{rem:generalization}. We also see that (\ref{eq:newresult}) reduces to (\ref{eq:relativeentequality}) in the limit $\alpha\rightarrow 0$. Since both equalities in the DPI's for the relative entropy and the sandwiched R\'enyi divergence lead to the recoverability statement, (\ref{eq:newresult}) is equivalent to (\ref{eq:relativeentequality}). We however cannot directly obtain (\ref{eq:relativeentequality}) from our analysis for values of $\alpha$ different from $\alpha=0$.

Before starting with the technical part of our paper, let us briefly comment on the topic of \emph{approximate} recovery: A seminal result by Fawzi and Renner \cite{fawzi2015quantum} and follow-up works \cite{Sutter2016a,wilde2015recoverability,Sutter2016b,Sutter2016,junge2018universal,carlen2020recovery,faulkner2020approximate} showed that when the DPI of the quantum relative entropy is close to being saturated, then $\rho$ may be recovered from $\Lambda(\rho)$ \emph{approximately} by a recovery-channel $\mc R_{\sigma,\Lambda}$ that at the same time \emph{exactly} recovers $\sigma$ from $\Lambda(\sigma)$. In this sense, recoverability is \emph{robust}. Very recently, Gao and Wilde \cite{gao2020recoverability} obtained results on robust recovery for optimized quantum $f$-divergences, which include the sandwiched R\'enyi divergence, also using the relative modular operator approach. These naturally contain the case of exact recovery as a special case, but their proof is much more demanding than the special case of equality that we treat in this work.

\emph{Notation.} We denote the set of bounded operators on a Hilbert space $\H$ as $\B(\H)$ and the set of unit-trace, positive (semi-)definite operators on $\H$ as $\P_+(\H)\ (\P(\H))$ and CPTP maps as $\l\in\mathrm{CPTP}(\H,\K)$. Given a Hilbert-space $\mc H$ and a density matrix $\rho$ with orthonormal eigenbasis $\{\ket{j}\}_j$, we define the vector $\ket{\id} := \sum_j \ket{j}\otimes\ket{j}$. To any operator $a\in \B(\H)$ we associate the vector $\ket{a} := a\otimes\id \ket{\id}$, providing an isomorphism between $\B(\H)$ with Hilbert-Schmidt inner product and $\H\otimes \H$: $\braket{a}{b}=\tr[a^\dagger b]$. 
The vector $\ket{\rho^\frac12}$ is called the \emph{canonical purification} of $\rho$. Since $\rho>0$, any vector $\ket{\Psi}\in\mc H\otimes \mc H$ may be written as $a_\Psi\otimes\mathbb I \ket{\rho^\frac12}$.  
If $\Lambda:\B(\H)\rightarrow\B(\K)$ is a super-operator, it induces an operator $\hat \Lambda$ on $\H\otimes\H$ as $\hat \Lambda\ket{a} = \ket{\Lambda(a)}$, where we assume implicitly that a choice of reference density matrix is also given for $\K$. In the following, we will often identify $\hat \Lambda$ and $\Lambda$.  
Finally, if we talk about bipartite systems consisting of two parts $A$ and $B$ with Hilbert space $\H_A$ and $\H_B$, respectively, we will often identify $A$ with $\H_A$ etc. Finally, we will oftentimes omit subscripts as labels for sub-systems from operators if they are clear from context.

\section{Preliminaries}\label{sec:prelim}
\subsection{R\'enyi divergences via the relative modular operator}
Denote by $L_\sigma$ and $R_{\rho^{-1}}$ the super-operators for left- and right-multiplication by $\sigma$ and $\rho^{-1}$, respectively. The relative modular super-operator is then defined as
\begin{equation}
	\Delta_{\sigma,\rho}:= L_\sigma R_{\rho^{-1}}.
\end{equation}
Petz used it to define the quasi-entropy, also known as the quantum $f$-divergence, as \cite{hiai2011quantum}
\begin{equation}
	S_f(\rho||\sigma) := \bra{\rho^\frac12} f(\Delta_{\sigma,\rho}) \ket{\rho^\frac12}
\end{equation}
where $f$ is an operator convex function and $\ket{\rho^\frac12}$ denotes the canonical purification of $\rho$. 
The DPI of this quantity follows from operator Jensen's inequality and the convexity of $f$.

Our analysis of the sandwiched R\'enyi divergences is based on the technique of relative modular opeartor \`a la Petz. It's useful to introduce it via briefly representing the proof for the relative entropy DPI in the special case where the quantum channel corresponds to the partial trace $\tr_B$. The relative entropy corresponds to the choice $f=-\log.$
The proof makes use of the following super-operators:
\begin{align}
	\Delta_{\sigma,\rho}^A (\cdot)  = \sigma_A \cdot\rho_A^{-1},\quad \Delta_{\sigma,\rho}^{AB}(\cdot) = \sigma_{AB}\cdot\rho_{AB}^{-1},\quad U(\cdot) =\,\cdot\, \rho_A^{-\frac12}\otimes \id_B\rho_{AB}^\frac12,\quad U^*(\cdot)= \tr_B[\,\cdot\, \rho_{AB}^\frac12] \rho_A^{-\frac12}\nonumber
\end{align}
where $^*$ denotes either the adjoint with respect to the Hilbert-Schmidt inner product or the adjoint on operators on Hilbert-space, depending on context. The operator $U$ is an isometry from $A$ to $AB$.
We therefore have $U^*U=\mathbf 1_A$ and define the projector $P:=UU^*$ on the system $AB$. Here, $\mathbf 1_A$ denotes the identity map on (not in) the algebra of operators on $A$. 
By construction
\begin{align}
U^*\Delta_{\sigma,\rho}^{AB}U=\Delta_{\sigma,\rho}^A. 
\end{align}
The quantum relative entropy then simply takes the form
\begin{align}
	-\mathrm{Tr}[\rho_A^\frac12\log(\Delta_{\sigma,\rho}^A)(\rho_A^\frac12)] = -\mathrm{Tr}[\rho_A^\frac12\log(\sigma_A)\rho_A^\frac12] + \mathrm{Tr}[\rho_A\log(\rho_A)] = D(\rho_A\|\sigma_A)
\end{align}
where we used that $\log(\Delta_{\sigma,\rho}^A)(a) = \log(\sigma_A)a - a\log(\rho_A)$.
We now make use of Jensen's operator inequality, which states that $f(U^*aU)\leq U^*f(a)U$ for a convex function $f$ and self-adjoint $a$ \footnote{This is not to be confused with the statement $f(UbU^*) = Uf(b)U^*$.}.
Setting $a=\Delta_{\sigma,\rho}^{AB}, f=-\log$, we then get the DPI:
\begin{equation}
\begin{aligned}
	D(\rho_A\|\sigma_A) &= -\mathrm{Tr}[\rho_A^\frac12\log(U^*\Delta_{\sigma,\rho}^{AB}U)(\rho_A^\frac12)] \leq -\mathrm{Tr}[\rho_A^\frac12U^*\log(\Delta_{\sigma,\rho}^{AB})U(\rho_A^\frac12)],\\
	&= -\mathrm{Tr}[\rho_{AB}^\frac12\log(\Delta_{\sigma,\rho}^{AB})(\rho_{AB}^\frac12)] = D(\rho_{AB}\| \sigma_{AB}).
\end{aligned}
\end{equation}
The only inequality used in the proof is the operator Jensen's inequality. In Appendix~\ref{app:jensenequality} we show that the equality condition for Jensen's inequality (which implies saturation of the DPI) directly gives rise to the  Petz recovery map. 

\subsection{Sandwiched R\'enyi divergence}
Let us now return to the sandwiched R\'enyi divergence. Since we wish to use the machinery of relative modular operators, it's desirable to have an expression for the sandwiched R\'enyi divergences based on relative modular operators. We would therefore like to rewrite $\D_\alpha$ using a variational formula that involves the relative modular operator. The Araki-Masuda norms of the $L_p$ spaces~\cite{araki1982positive} will be useful. For a density operator $\rho$ and a reference density operator $\sigma$, the Araki-Masuda $p$-norms are defined as
\begin{equation}
\begin{aligned}
\norm{\rho}_{p,\sigma} := &\sup_{\omega\in\P_+(\H)}\norm{\Delta^{\frac1p-\frac12}_{\sigma,\omega}\ket{\rho^\frac12}}_2, \,\,\,p\in(2,\infty) ,\\
\norm{\rho}_{p,\sigma} := &\inf_{\omega\in\P_+(\H)}\norm{\Delta^{\frac1p-\frac12}_{\sigma,\omega}\ket{\rho^\frac12}}_2, \,\,\,p\in[1,2).
\end{aligned}
\end{equation}
One can evaluate the maximizers and minimizers explicitly,
\begin{equation}
\begin{aligned}
\norm{\rho}^2_{p,\sigma}=&\sup_{\omega\in\P_+(\H)}\bra{\rho^\frac12}\Delta^{\frac2p-1}_{\sigma , \omega}\ket{\rho^\frac12}=\sup_{\omega\in\P_+(\H)}\tr \rho^\frac12\sigma^{\frac2p-1}\rho^\frac12\omega^{1-\frac2p} = \sup_{\omega\in\P_+(\H)}\tr Y \omega^{1-\frac2p},\,\,\,p\in(2,\infty),\\
\norm{\rho}^2_{p,\sigma}=&\inf_{\omega\in\P_+(\H)}\bra{\rho^\frac12}\Delta^{\frac2p-1}_{\sigma , \omega}\ket{\rho^\frac12}=\inf_{\omega\in\P_+(\H)}\tr \rho^\frac12\sigma^{\frac2p-1}\rho^\frac12\omega^{1-\frac2p} = \inf_{\omega\in\P_+(\H)}\tr Y \omega^{1-\frac2p}  ,\,\,\,\,\,p\in[1,2).
\end{aligned}
\end{equation}
where we have defined $Y=\rho^\frac12\sigma^{\frac2p-1}\rho^\frac12>0$. The following Lemma shows that the optimizer is unique.
\begin{lem}
For $p\in [1,2)$ and $(2,\infty),$ the function $\omega\mapsto \tr Y \omega^{1-\frac2p}$ is strictly convex and concave, respectively. 
\end{lem}
\begin{proof} Since $Y>0,$ we can split $Y=c\id+Y'$ where $c$ is the smallest eigenvalue of $Y$ and $Y'\geq 0$. We have
\begin{equation}
\tr Y \omega^{1-\frac2p} = \tr Y'\omega^{1-\frac2p} + c\tr\,\omega^{1-\frac2p}.
\end{equation}
	The first term can be written as $\tr Y^\frac12 \omega^{1-\frac2p} Y^\frac12$. It is operator convex and concave as the function $\omega\mapsto\omega^{1-\frac2p}$ is strictly operator convex and concave for $p\in [1,2)$ and $(2,\infty)$ respectively. The second term has the form $\tr f(X)$ where $f$ is a strictly convex (concave) function, so it is strictly convex and concave for $p\in [1,2)$ and $(2,\infty)$, respectively \cite{carlen2010trace}. Therefore, the sum of the two terms are strictly convex and concave for $p\in [1,2)$ and $(2,\infty)$ respectively. 
\end{proof}
The unique optimizer $\omega_*$ for $\norm{\rho}^2_{p,\sigma}$ is given by \cite{muller2013quantum}:
\begin{equation}\label{eq:optimizer}
\omega_*=(\tr Y^\frac{p}{2})^{-1}Y^\frac{p}{2}.
\end{equation}
Hence, the Araki-Masuda $p$-norms reduce to the $\sigma$-weighted $p$-norms~\cite{berta2018renyi}:
\begin{equation}
\norm{\rho}_{p,\sigma} = (\tr Y^\frac{p}{2})^\frac1p =\norm{Y^\frac12}_p = \norm{\sigma^{\frac1p-\frac12}\rho^\frac12}_p,
\end{equation}
where $\norm{\cdot}_p$ are the Schatten $p$-norms. 

In terms of the Araki-Masuda $p$-norms, the sandwiched quantum R\'enyi divergences can hence be defined as 
\begin{align}
	\D_\alpha(\rho||\sigma)&:=\frac{2}{\alpha}\log\norm{\rho}_{\frac{2}{1-\alpha},\sigma}, \,\,\,\alpha\in[-1,0)\cup (0,1).
\end{align}
Equivalently, we can write:
\begin{equation}\label{eq:definitionRenyi}
\begin{aligned}
\D_\alpha(\rho||\sigma)&:=\frac{1}{\alpha}\inf_{\omega\in\P_+(\H)}\log\bra{\rho^\frac12}\Delta_{\sigma,\omega}^{-\alpha}\ket{\rho^\frac12}\,,\quad\alpha\in[-1,0),\\
\D_\alpha(\rho||\sigma)&:=\frac{1}{\alpha}\sup_{\omega\in\P_+(\H)}\log\bra{\rho^\frac12}\Delta_{\sigma,\omega}^{-\alpha}\ket{\rho^\frac12}\,,\quad\,\alpha\in(0,1).
\end{aligned}
\end{equation}

Note that this is not a quantum $f$-divergence but rather an instance of the optimized quantum $f$-divergence introduced in \cite{wilde2018optimized}. If one removes the supremum and infimum in (\ref{eq:definitionRenyi}), then the above expressions define the Petz quantum R\'enyi divergences $\overline{D}_\alpha$:
\begin{equation}
\overline{D}_\alpha(\rho||\sigma):=\frac{1}{\alpha}\log\bra{\rho^\frac12}\Delta_{\sigma,\rho}^{-\alpha}\ket{\rho^\frac12}, \,\,\,\alpha\in[-1,0)\cup (0,1).
\end{equation}
The Petz quantum R\'enyi divergence is the exponent of the $f$-divergence $S_f(\rho||\sigma)$ with the operator convex function $f(x)=\text{sign}(\alpha)x^{-\alpha}$. To see why this power function function is operator convex, we adopt the following integral representation \cite{bhatia2013matrix}:
\begin{equation}\label{eq:integral}
\begin{aligned}
x^\alpha =& \frac{\sin(\pi\alpha)}{\pi}\int^\infty_0\dd t\, t^\alpha\left(\frac{1}{t}-\frac{1}{t+x}\right), \,\,\,\alpha\in(0,1),\\
	x^\alpha =& -\frac{\sin(\pi\alpha)}{\pi}\int^\infty_0\dd t\, \frac{t^\alpha}{t+x} \;\;,\;\;\;\;\;\;\;\;\;\;\;\;\alpha\in(-1,0).
\end{aligned}
\end{equation}

Since the function $f(x)=(x+t)^{-1}$ is operator convex and operator anti-monotone for $t\geq 0,$ the same holds for $\text{sign}(\alpha)x^{-\alpha}$ after integration. Note that the above formulae do not hold for the boundary cases $\alpha=\pm 1$. However, they are obviously operator convex/concave as well. The integral representations will be useful later for establishing the equality condition.

\subsection{The data-processing inequality}

Now let's turn to the data processing inequality and w.l.o.g. restrict the map to the partial trace. We will later come back to the general case, which follows from Stinepspring's dilation theorem. The DPI for the Petz quantum R\'enyi divergences with $\alpha\in [-1,0)\cup (0,1)$ follows directly from the operator convexity of the power function $f(x)=\text{sign}(\alpha)x^{-\alpha}$.

The DPI of the sandwiched R\'enyi divergences $\tilde{D}_\alpha$ involves more steps. 
Consider an arbitray vector $\ket{\omega^\frac12}_{AA'}\in \mathcal{H}_{AA'}$ which can be written as $\ket{a\rho^\frac12}_{AA'}$, for some $a\in\mathcal{B}(\mathcal{H}_A)$ (for simplicity of notation we here simply write $\rho$ and $\omega$ instead of $\rho_A$ and $\omega_A$). Here, as before, $\ket{\omega^\frac12}_{AA'}$ denotes the canonical purification of $\omega$ with the purifying system $A'$. Similarly, we write $\ket{\omega^\frac12}_{AA'BB'}=\ket{b\rho^\frac12}_{AA'BB'}$ for some $b\in\mathcal{B}(\mathcal{H}_{AB})$ (again identifying $\rho$ and $\rho_{AB}$ as well as $\omega$ and $\omega_{AB}$ from context). We can then write the relative modular operator $\Delta^A_{\sigma,\omega}$ as $\Delta^A_{\sigma,\rho^{1/2} |a|^2 \rho^{1/2}}.$ One can readily check that
\begin{equation}\label{eq:modular}
	U^*\Delta^{AB}_{\sigma_{AB},\rho_{AB}^{1/2} (|a|^2\otimes\mathbb I) \rho_{AB}^{1/2}}U=\Delta^A_{\sigma_{A},\rho_{A}^{1/2} |a|^2 \rho_{A}^{1/2}}\,.
\end{equation} 

 Let the maximizer for $\sup_a\bra{\rho^\frac12}\Delta^{A -\alpha}_{\sigma,\rho^{1/2}a^2\rho^{1/2}}\ket{\rho^\frac12}$ be $a_*.$ According to (\ref{eq:optimizer}), $a_*^2$ is given by the following Hermitian operator
\begin{equation}
	a_*^2= \frac{\rho_A^{-\frac12}\left(\rho_A^\frac12\sigma_A^{-\alpha}\rho_A^\frac12\right)^\frac{1}{1-\alpha}\rho_A^{-\frac12} }{ \tr\left(\rho_A^\frac12\sigma_A^{-\alpha}\rho_A^\frac12\right)^\frac{1}{1-\alpha}}.
\end{equation} 

First we apply the Jensen's operator inequality, the convexity of $f=\text{sign}(\alpha)x^{-\alpha}$ and $(\ref{eq:modular})$:
\begin{equation}
	\text{sign}(\alpha)\bra{\rho_A^\frac12}\Delta^{A -\alpha}_{\sigma_A,\rho_A^{1/2}a_*^2\rho_A^{1/2}}\ket{\rho_A^\frac12}\leq \text{sign}(\alpha)\bra{\rho_{AB}^\frac12}\Delta^{AB -\alpha}_{\sigma_{AB},\rho_{AB}^{1/2}(a_*^2\otimes \mathbb I)\rho_{AB}^{1/2}}\ket{\rho_{AB}^\frac12}.
\end{equation} 
Since vectors of the form $\ket{(a\otimes \mathbb I)\rho^\frac12}_{AA'BB'}$ constitute a restricted set of vectors in $\mathcal{H}_{AA'BB'},$ we have 
\begin{equation}\label{ineq:DPI}
\begin{aligned}
\text{sign}(\alpha)\bra{\rho_A^\frac12}\Delta^{A -\alpha}_{\sigma_A,\rho_A^{1/2}a_*^2\rho_A^{1/2}}\ket{\rho_A^\frac12}&\leq \text{sign}(\alpha)\bra{\rho_{AB}^\frac12}\Delta^{AB -\alpha}_{\sigma_{AB},\rho_{AB}^{1/2}(a_*^2\otimes \mathbb I)\rho_{AB}^{1/2}}\ket{\rho_{AB}^\frac12},\\ &\leq \text{sign}(\alpha)\bra{\rho_{AB}^\frac12}\Delta^{AB -\alpha}_{\sigma_{AB},\rho_{AB}^{1/2}b_*^2\rho_{AB}^{1/2}}\ket{\rho_{AB}^\frac12}.
\end{aligned}
\end{equation} 
where in the last step, we use the fact that $\rho^\frac12_{AB}b_*^2\rho^\frac12_{AB}$ is the supremum (infimum) in  (\ref{eq:definitionRenyi}) for $\alpha>0 \, (\alpha<0)$ respectively.
It then follows that
\begin{equation}
\begin{aligned}
	\tilde{D}_\alpha(\rho_A||\sigma_A)=\frac1{\alpha}\log\bra{\rho_A^\frac12}\Delta^{A -\alpha}_{\sigma_A,\rho_A^{1/2}a_*^2\rho_A^{1/2}}\ket{\rho_A^\frac12}
	&\leq\frac1{\alpha}\log \bra{\rho_{AB}^\frac12}\Delta^{AB -\alpha}_{\sigma_{AB},\rho_{AB}^{1/2}(a_*^2\otimes\mathbb{I})\rho_{AB}^{1/2}}\ket{\rho_{AB}^\frac12},\\
&\leq \frac1{\alpha}\log \bra{\rho_{AB}^\frac12}\Delta^{AB -\alpha}_{\sigma_{AB},\rho_{AB}^{1/2}b_*^2\rho_{AB}^{1/2}}\ket{\rho_{AB}^\frac12},\\
	&=\tilde D_\alpha(\rho_{AB}||\sigma_{AB}),
\end{aligned}
\end{equation} 
where 
\begin{equation}
	b_*^2=\frac {\rho_{AB}^{-\frac12}\left(\rho_{AB}^\frac12\sigma_{AB}^{-\alpha}\rho_{AB}^\frac12\right)^\frac{1}{1-\alpha}\rho_{AB}^{-\frac12}}{\tr\left(\rho_{AB}^\frac12\sigma_{AB}^{-\alpha}\rho_{AB}^\frac12\right)^\frac{1}{1-\alpha}}.
\end{equation}
Note that the denominator is simply given by $\exp(\frac{\alpha}{1-\alpha}\D_\alpha(\rho_{AB}||\sigma_{AB}))$.

\subsection{The geometric mean}
It will turn out to be convenient to use the geometric mean to compactify the equality conditions, so we shall briefly introduce some facts about it. The $\lambda$-weighted geometric mean between two positive matrices is defined as~\cite{bhatia2009positive}
\begin{equation}
A\sharp_\lambda B = A^\frac12(A^{-\frac12}BA^{-\frac12})^\lambda A^\frac12 .
\end{equation}
Usually, the geometric mean is defined for $\lambda\in [0,1]$, but the above RHS is well-defined for any $\lambda\in\mathbb{R}$. We shall use the notation for any real $\lambda$. 
The following properties will be useful later:
\begin{lem}\label{lem:mean}
For $A,B >0, \lambda\in\mathbb{R},$ the following are true:
	\begin{enumerate}
		\item	$A\sharp_\lambda B = B\sharp_{1-\lambda} A$,
		\item $(A\sharp_\lambda B)^{-1}=A^{-1}\sharp_\lambda B^{-1}$,
		\item $A\sharp_\lambda B = A (A^{-1}B)^{\lambda} = (AB^{-1})^{1-\lambda}B$.
	\end{enumerate}
\end{lem}
The geometric mean is particularly useful for us, since the optimizers $a_*^2,b_*^2$ may be written as geometric means with $\lambda=\frac{1}{1-\alpha} (=n)$:
\begin{equation}\label{eq:maximizers}
a_*^{-2}\propto\rho_A\sharp_{\frac{1}{1-\alpha}} \sigma_A^{\alpha},\quad b_*^{-2}\propto\rho_{AB}\sharp_{\frac{1}{1-\alpha}} \sigma_{AB}^{\alpha} .
\end{equation} 

\section{Equality conditions for the $\D_\alpha$ DPI}
We are now in a position to study the equality conditions for the sandwiched R\'enyi divergence.
Let's first simplify the situation by restricting to the partial trace map $\Lambda=\tr_B$.
We will now state the three theorems from the introduction in this simplified setting and give their proofs. 
First, using the geometric mean, we can state the result by Leditzky, Rouz\'e and Datta as follows.
\begin{manualtheorem}{1'}[Leditzky, Rouz\'e and Datta]\label{theorem:equality'}
	\emph{For $\alpha\in[-1,0)\cup(0,1)$, density matrices $\sigma_{AB},\rho_{AB},$ $\D_\alpha(\sigma_A||\rho_A) = \D_\alpha(\sigma_{AB}||\rho_{AB})$ if and only if}
\begin{equation}\label{eq:oldresult}
\rho_A\sharp_{\frac{1}{1-\alpha}} \sigma_A^{\alpha}\otimes\id_B = \rho_{AB}\sharp_{\frac{1}{1-\alpha}} \sigma_{AB}^{\alpha}.
\end{equation}
\end{manualtheorem}
To obtain Theorem~\ref{theorem:equality}, we simply take the inverse on both sides of (\ref{eq:oldresult}) and swap $\rho^{-1}$ and $\sigma^{-\alpha}$. Using Lemma \ref{lem:mean}, we have
\begin{equation}\label{eq:remark}
\sigma_A^{-\alpha}\sharp_{\frac{-\alpha}{1-\alpha}} \rho_A^{-1}=\sigma_{AB}^{-\alpha}\sharp_{\frac{-\alpha}{1-\alpha}} \rho_{AB}^{-1}\,,
\end{equation}
which is equivalent to (\ref{eq:equalitycondition1}) for $\Lambda=\tr_B.$
The second theorem concerns the recoverability:
\begin{manualtheorem}{2'}[Jen\v cov\'a]\label{theorem:recoverability'}
	\emph{For $\alpha\in(-1,0)\cup(0,1)$, density matrices $\sigma_{AB},\rho_{AB}$ and the quantum channel $\Lambda=\tr_B$ we have $\D_\alpha(\rho_{AB}||\sigma_{AB}) = \D_\alpha(\rho_A||\sigma_A)$ if and only if $\R_{\sigma,\tr_B}(\rho_{A})=\rho_{AB}$.}
\end{manualtheorem}
The third theorem provides our new equality conditions:
\begin{manualtheorem}{3'}\label{theorem:newequality'}
	\emph{For $\alpha\in(-1,0)\cup(0,1)$, density matrices $\sigma_{AB},\rho_{AB},$ $\D_\alpha(\sigma_A||\rho_A) = \D_\alpha(\sigma_{AB}||\rho_{AB})$ if and only if}
\begin{equation}\label{eq:newresult'}
\sigma_A^\beta(\rho_A\sigma_A^{-\alpha})^{\frac{\beta}{\alpha-1}}\otimes\id_B = \sigma_{AB}^\beta(\rho_{AB}\sigma_{AB}^{-\alpha})^{\frac{\beta}{\alpha-1}},\quad \forall \beta\in \mathbb C.
\end{equation}
\end{manualtheorem}

\begin{rem}\label{rem:generalization}
	Theorem \ref{theorem:newequality'} implies Theorem \ref{theorem:equality'}: To see that, pick $\beta=-\alpha$ and use property $\it{3}$ in Lemma \ref{lem:mean} to get (\ref{eq:remark}):
	\begin{align}
		\sigma^{-\alpha} (\rho \sigma^{-\alpha})^{\frac{-\alpha}{\alpha-1}} = \sigma^{-\alpha}(\sigma^\alpha \rho^{-1})^{\frac{-\alpha}{1-\alpha}} = \sigma^{-\alpha}\sharp_{\frac{-\alpha}{1-\alpha}} \rho^{-1}.
	\end{align}
\end{rem}

All three theorems rely on the simple observation that the DPI saturation amounts to setting equalities in (\ref{ineq:DPI}):
\begin{equation}\label{eq:equality}
	\bra{\rho_A^\frac12}\Delta^{A -\alpha}_{\sigma_A,\rho_A^{1/2}a_*^2\rho_A^{1/2}}\ket{\rho_A^\frac12} \stackrel{\circled{1}}{=} \bra{\rho_{AB}^\frac12}\Delta^{AB -\alpha}_{\sigma_{AB},\rho_{AB}^{1/2}(a_*^2\otimes \mathbb{I})\rho_{AB}^{1/2}}\ket{\rho_{AB}^\frac12} \stackrel{\circled{2}}{=} \bra{\rho_{AB}^\frac12}\Delta^{AB -\alpha}_{\sigma_{AB},\rho_{AB}^{1/2}b_*^2\rho_{AB}^{1/2}}\ket{\rho_{AB}^\frac12}\,.
\end{equation}
The proof of Theorem~\ref{theorem:equality'} will only use condition $\circled{2}$, while the proof of the recoverability statement Theorem~\ref{theorem:recoverability'} only uses condition $\circled{1}$. 
This means that \circled{1} and \circled{2} imply each other in a nontrivial way. 
Theorem~\ref{theorem:newequality'} uses both $\circled{1}$ and $\circled{2}$.

\begin{proof}[Proof of Theorem~\ref{theorem:equality'}]
	The proof is rather straightforward. For the necessary condition, observe that the second equality \circled{2} above implies $a_*^2\otimes \mathbb{I}=b_*^2$ as the optimizer is unique according to (\ref{eq:optimizer}).  
	According to (\ref{eq:maximizers}), it gives (\ref{eq:oldresult}), since the scalar pre-factors are identical by the assumption $\D_\alpha(\rho_A||\sigma_A) = \D_\alpha(\rho_{AB}||\sigma_{AB})$. To see sufficiency, we can simply left-multiply with $\rho_{AB}^{-1}$ and take the trace.
\end{proof}
To proof Theorem~\ref{theorem:recoverability'}, we shall follow the argument \`a la Petz on the equality condition of the relative entropy~\cite{petz2003monotonicity}.
\begin{proof}[Proof of Theorem \ref{theorem:recoverability'}]
Since $f(x)=(x+t)^{-1}, t\geq 0$ is operator convex, we have
	\begin{equation}\label{eq:petzargument0}
		X_t := U^*(\Delta_{\sigma_{AB}, \rho_{AB}^{1/2}(a_*^2\otimes \mathbb{I})\rho_{AB}^{1/2}}+t)^{-1} U - (\Delta_{\sigma_A, \rho_A^{1/2}a_*^2\rho_A^{1/2}}+t)^{-1}\geq 0 ,\,\,\,\forall t.
\end{equation}
We now first show that
	\begin{equation}\label{eq:petzargument1}
	U^*(\Delta_{\sigma_{AB}, \rho_{AB}^{1/2}(a_*^2\otimes\mathbb{I})\rho_{AB}^{1/2}}+t)^{-1}\ket{\rho_{AB}^\frac12} = (\Delta_{\sigma_A, \rho_A^{1/2}a_*^2\rho_A^{1/2}}+t)^{-1} \ket{\rho^\frac12_A},
\end{equation}
	where $\ket{\rho_{AB}^\frac12}=U\ket{\rho_{A}^\frac12}$.
	To see this, first note that the DPI equality and the integral representations (\ref{eq:integral}) give
\begin{align}
	\int_0^\infty \mathrm{d}t\, t^\alpha \bra{\rho_A^\frac12}X_t \ket{\rho_A^\frac12} = 0.
\end{align}
	Since $X_t\geq 0$, the integrand is $\geq 0$ and hence must vanish for almost all $t\geq 0$ (and hence for all $t\geq 0$ by continuity). But since $X_t\geq 0$, this implies $X_t\ket{\rho_A^\frac12}=0$ for all $t\geq 0$, which is \eqref{eq:petzargument1}.
	Now acting with $U$ on the left on \eqref{eq:petzargument1} gives
\begin{equation}\label{eq:petzargument2}
	UU^*(\Delta_{\sigma_{AB}, \rho_{AB}^{1/2}(a_*^2\otimes \mathbb{I})\rho_{AB}^{1/2}}+t)^{-1}\ket{\rho_{AB}^\frac12} = U(\Delta_{\sigma_A, \rho_A^{1/2}a_*^2\rho_A^{1/2}}+t)^{-1} \ket{\rho^\frac12_A}.
\end{equation}
We now show that the LHS is invariant under the projection $P=UU^*$, i.e.,
	\begin{align}
		P(\Delta_{\sigma_{AB}, \rho_{AB}^{1/2}(a_*^2\otimes \mathbb{I})\rho_{AB}^{1/2}}+t)^{-1}\ket{\rho_{AB}^\frac12} = (\Delta_{\sigma_{AB}, \rho_{AB}^{1/2}(a_*^2\otimes \mathbb{I})\rho_{AB}^{1/2}}+t)^{-1}\ket{\rho_{AB}^\frac12}. \label{eq:petzargument_projection}
	\end{align}
	To see this, first take the trace of the square of the modulus of (\ref{eq:petzargument1}) to get
\begin{align}
	\left\langle(\Delta_{\sigma_{AB}, \rho_{AB}^{1/2}(a_*^2\otimes \mathbb{I})\rho_{AB}^{1/2}}+t)^{-1}(\rho_{AB}^\frac12)\right|P &\left|(\Delta_{\sigma_{AB}, \rho_{AB}^{1/2}(a_*^2\otimes\mathbb{I})\rho_{AB}^{1/2}}+t)^{-1}(\rho_{AB}^\frac12)\right\rangle\nonumber \\ &= \braket{(\Delta_{\sigma_A, \rho_A^{1/2}a_*^2\rho_A^{1/2}}+t)^{-2} (\rho^\frac12_A)}{\rho^\frac12_A}, \nonumber
\end{align}
where we used that $(\Delta_{\sigma_A, \rho_A^{1/2}a_*^2\rho_A^{1/2}}+t)^{-1}$ is Hermitian.
Now consider the derivative of (\ref{eq:petzargument1}) with respect to $t$,
\begin{equation}\label{eq:petzargument3}
	U^*(\Delta_{\sigma_{AB}, \rho_{AB}^{1/2}(a_*^2\otimes\mathbb{I})\rho_{AB}^{1/2}}+t)^{-2}(\rho_{AB}^\frac12) = (\Delta_{\sigma_A, \rho_A^{1/2}a_*^2\rho_A^{1/2}}+t)^{-2} (\rho^\frac12_A),
\end{equation}
	and use it on the RHS to get
\begin{align}
	\left\langle(\Delta_{\sigma_{AB}, \rho_{AB}^{1/2}(a_*^2\otimes \mathbb{I})\rho_{AB}^{1/2}}+t)^{-1}(\rho_{AB}^\frac12)\right| P &\left|(\Delta_{\sigma_{AB}, \rho_{AB}^{1/2}(a_*^2\otimes \mathbb{I})\rho_{AB}^{1/2}}+t)^{-1}(\rho_{AB}^\frac12)\right\rangle\nonumber \\ 
	&= \braket{U^*(\Delta_{\sigma_{AB}, \rho_{AB}^{1/2}(a_*^2\otimes\mathbb{I})\rho_{AB}^{1/2}}+t)^{-2}(\rho_{AB}^\frac12)}{\rho^\frac12_A}.
\end{align}
	We can now use that $U(\rho_A^{\frac12})= \rho_{AB}^{\frac12}$ to find
	\begin{align}
		\bra{(\Delta_{\sigma_{AB}, \rho_{AB}^{1/2}(a_*^2\otimes \mathbb{I})\rho_{AB}^{1/2}}+t)^{-1}(\rho_{AB}^\frac12)}(\mathbf 1-P)\ket{(\Delta_{\sigma_{AB}, \rho_{AB}^{1/2}(a_*^2\otimes \mathbb{I})\rho_{AB}^{1/2}}+t)^{-1}(\rho_{AB}^\frac12)} =& \,\,0\,.
\end{align}
	Since $P$ is an orthogonal projection, this implies (\ref{eq:petzargument_projection}).
We thus see that the LHS of (\ref{eq:petzargument2}) is invariant under the projection $P$ and obtain
\begin{equation}
	U(\Delta_{\sigma_A, \rho_A^{1/2}a_*^2\rho_A^{1/2}}+t)^{-1} (\rho^\frac12_A)=(\Delta_{\sigma_{AB}, \rho_{AB}^{1/2}(a_*^2\otimes \mathbb{I})\rho_{AB}^{1/2}}+t)^{-1}(\rho_{AB}^\frac12).
\end{equation}
Integrating this equation using the integral representations (\ref{eq:integral}) again gives
\begin{equation}\label{eq:pickup}
	U\Delta^{\beta}_{\sigma_A, \rho_A^{1/2}a_*^2\rho_A^{1/2}}(\rho_A^\frac12) = \Delta^{\beta}_{\sigma_{AB},\rho_{AB}^{1/2}(a_*^2\otimes\id)\rho_{AB}^{1/2}}(\rho_{AB}^\frac12),\,\,\,\forall \beta\in (-1,0)\cup (0,1).
\end{equation}
Let us choose $\beta=-\frac12$ and expand (\ref{eq:pickup}):
\begin{equation}
 \sigma_{A}^{-\frac12} \rho_{A}^\frac12 (\rho_A^\frac12\sigma_A^{-\alpha}\rho_A^\frac12)^\frac{1}{2-2\alpha}\rho_A^{-\frac12}\otimes \id 
	= \sigma^{-\frac12}_{AB}\rho^{\frac12}_{AB} \left(\rho_{AB}^\frac12\left(\rho_{A}^{-\frac12}(\rho_{A}^\frac12\sigma_{A}^{-\alpha}\rho_{A}^\frac12)^\frac{1}{1-\alpha}\rho_{A}^{-\frac12}\otimes \mathbb{I}\right)\rho_{AB}^\frac12\right)^\frac12\rho_{AB}^{-\frac12}.\nonumber
\end{equation}
Now define $L:=\left((\rho_{A}^\frac12\sigma_{A}^{-\alpha}\rho_{A}^\frac12)^\frac{1}{2-2\alpha}\rho_{A}^{-\frac12}\otimes \mathbb{I}\right)\rho_{AB}^\frac12$. Then we can write the above equation as
\begin{equation}
\sigma_{A}^{-\frac12} \rho_{A}^\frac12 \otimes \id_B L|L|^{-1} = \sigma^{-\frac12}_{AB}\rho^{\frac12}_{AB},
\end{equation}
where $|L|:=(L^*L)^\frac12$. Taking the Hermitian square, we obtain
\begin{equation}
\sigma_{A}^{-\frac12} \rho_{A}^\frac12 \otimes \id_B ( L|L|^{-2}L^*) \rho_{A}^\frac12  \sigma_{A}^{-\frac12} \otimes \id_B = \sigma^{-\frac12}_{AB}\rho_{AB} \sigma^{-\frac12}_{AB}\,,
\end{equation}
which simplifies to
\begin{equation}
\sigma^{\frac12}_{AB}\left( \sigma_{A}^{-\frac12} \rho_{A}\sigma_{A}^{-\frac12} \otimes \id_B\right)\sigma^{\frac12}_{AB}  = \rho_{AB}\,.
\end{equation}
So we have perfect Petz recovery $\R_{\sigma,\tr_B}(\rho_{AB})=\rho_{AB}$.
\end{proof}
Finally, let us now combine the equality conditions separately found for equalities \circled{1} and \circled{2} and see what we obtain.
\begin{proof}[Proof of Theorem \ref{theorem:newequality'}]
We directly pick up from the previous proof at the step (\ref{eq:pickup}), and replace $a_*^2$ with $b_*^2$. For any $\beta\in (-1,0)\cup (0,1), \alpha\in (-1,0)\cup (0,1)$, we then have have
\begin{equation}\label{eq:mainlemma}
U\Delta^{\beta}_{\sigma_A, \rho_A^{1/2}a_*^2\rho_A^{1/2}}(\rho_A^\frac12) = \Delta^{\beta}_{\sigma_{AB},\rho_{AB}^{1/2}b_*^2\rho_{AB}^{1/2}}(\rho_{AB}^\frac12),
\end{equation}
which leads to
\begin{equation}
\sigma_{A}^\beta \rho_{A}^\frac12 (\rho_A^\frac12\sigma_A^{-\alpha}\rho_A^\frac12)^\frac{\beta}{\alpha-1}\rho_A^{-\frac12}\otimes \id_B = \sigma^\beta_{AB}\rho^\frac12_{AB} (\rho_{AB}^\frac12\sigma_{AB}^{-\alpha}\rho_{AB}^\frac12)^\frac{\beta}{\alpha-1}\rho_{AB}^{-\frac12}.
\end{equation}
We can re-write this equation using the geometric mean,
\begin{align}
\sigma_{A}^\beta \rho_{A}(\rho_{A}^{-1}\sharp_\frac{\beta}{\alpha-1}\sigma_A^{-\alpha})\otimes \id_B =& \sigma^\beta_{AB}\rho_{AB}(\rho^{-1}_{AB}\sharp_\frac{\beta}{\alpha-1}\sigma_{AB}^{-\alpha}),
\end{align}
	and then use the property $A\sharp_\lambda B = A(A^{-1}B)^\lambda$ to get
\begin{align}
 \sigma_{A}^\beta (\rho_{A}\sigma_A^{-\alpha})^\frac{\beta}{\alpha-1}\otimes \id_B =& \sigma^\beta_{AB}(\rho_{AB}\sigma_{AB}^{-\alpha})^\frac{\beta}{\alpha-1},
\end{align}
which is the desired condition.
Note that for $X>0,\, \beta\rightarrow X^\beta$ is an entire function. Therefore, (\ref{eq:mainlemma}) and the final condition in fact hold for any $\beta\in\mathbb{C}$. This shows the necessary condition. To see sufficiency, we choose $\beta=-1$, left-multiply with $\sigma_{AB}$ and the take the trace. 
\end{proof}

Finally, we can use the standard Stinespring dilation argument to show that Theorem \ref{theorem:newequality'} implies Theorem \ref{theorem:newequality}. We follow \cite{leditzky2017data} where they show Theorem \ref{theorem:equality'} implies Theorem \ref{theorem:equality}.
\begin{proof}[Proof of Theorem \ref{theorem:newequality}] 
For any CPTP map $\l:\B(\H)\rightarrow\B(\K)$, the Stinespring dilation theorem allows us to dilate $\l$ to an isometry $V:\H \rightarrow \K\otimes\H'$ such that for all $\rho\in\B(\H)$,
\begin{equation}
\l(\rho)=\tr_{\H'} V \rho V^*.
\end{equation}

Since the divergence is invariant under conjugation by an isometry, the DPI for any channel $\l$ can be traced down to the partial trace map $\tr_{\H'}$. When the DPI is saturated, our theorem \ref{theorem:newequality'} implies that 
\begin{equation}
\begin{aligned}
	\l(\sigma)^\beta (\l(\rho)\l(\sigma)^{-\alpha})^{\frac{\beta}{\alpha-1}}\otimes\id_{\H'} =&\,V\sigma^\beta V^*[V\rho V^*(V\sigma V^*)^{-\alpha}]^{\frac{\beta}{\alpha-1}}\,,\\
=&\, V \sigma^\beta (\rho\sigma^{-\alpha})^{\frac{\beta}{\alpha-1}}V^*\,,
\end{aligned}
\end{equation}
where we used that $f(V X V^*) = Vf(X)V^*$ for any function $f$, operator $X$ and isometry $V$.
The adjoint channel $\l^*:\B(\K)\rightarrow \B(\H)$ is given by $\l^*(a) = V^* a\otimes \id_{\H'} V$. Applying $V^* \cdot V$ to both sides of the above equality, we therefore get
\begin{equation}
\l^*\left(\l(\sigma)^\beta [\l(\rho)\l(\sigma)^{-\alpha}]^{\frac{\beta}{\alpha-1}}\right) = \sigma^\beta(\rho\sigma^{-\alpha})^{\frac{\beta}{\alpha-1}}\,,
\end{equation}
where we use that $V$ is an isometry: $V^*V=\mathbb I$. 
\end{proof}
Similarly, the exactly same Stinespring dilation gives Theorem \ref{theorem:recoverability} from \ref{theorem:recoverability'}.

\section{Discussion}
We have provided a unified treatment of the equality conditions for the DPI of the sandwiched R\'enyi divergences using the relative modular operator approach, emphasizing the role of the two conditions $\circled{1}$ and $\circled{2}$ in (\ref{eq:equality}) to obtain previously known results. While the two conditions lead us to different results, they are logically equivalent.
For future work, it may be interesting to study how the two (in-)equalities behave once we go away from the case of exact saturation of the DPI. 
We expect that the two inequalities then behave differently. As mentioned in the introduction, recently Gao and Wilde studied the case of robust recovery for optimized $f$-divergences \cite{gao2020recoverability}. It would be interesting to understand the connection between the two (in-)equalities $\circled{1}$ and $\circled{2}$ and their results on approximate recoverability.  

Let us now come back to the case of equality. Recently, similar equality conditions were independently proved for the more general $\alpha-z$ R\'enyi divergences \cite{chehade2020saturating,zhang2020equality}. These conditions are slightly different but they both reduce to Theorem \ref{theorem:equality} for the sandwiched R\'enyi divergence case $\alpha=z.$ Our result suggests that there are perhaps more equality conditions to be discovered for the $\alpha-z$ R\'enyi divergences. However, our approach via the relative modular operator might not be useful for these $\alpha-z$ generalizations as they generally cannot be formulated as instances of optimized quantum $f$-divergences. Of course, insights about equality conditions might nevertheless help us better understand the recoverability perspectives on these $\alpha-z$ R\'enyi divergences.\\

\textbf{Acknowlegdements.} We would like to thank Joe Renes and David Sutter for comments and discussions. This work is supported by the Swiss National Science Foundation via the National Center for Competence in Research ``QSIT", and by the Air Force Office of Scientific Research (AFOSR) via grant FA9550-19-1-0202. 

\bibliographystyle{JHEP}
\bibliography{dpi}

\providecommand{\href}[2]{#2}\begingroup\raggedright\begin{thebibliography}{10}

\bibitem{petz1986quasi}
D.~Petz, \emph{Quasi-entropies for finite quantum systems},
  \href{https://doi.org/10.1016/0034-4877(86)90067-4}{\emph{Reports on
  mathematical physics} {\bfseries 23} (1986) 57}.

\bibitem{hiai2011quantum}
F.~Hiai, M.~Mosonyi, D.~Petz and C.~B{\'e}ny, \emph{Quantum f-divergences and
  error correction},
  \href{https://doi.org/10.1142/S0129055X11004412}{\emph{Reviews in
  Mathematical Physics} {\bfseries 23} (2011) 691}.

\bibitem{muller2013quantum}
M.~M{\"u}ller-Lennert, F.~Dupuis, O.~Szehr, S.~Fehr and M.~Tomamichel, \emph{On
  quantum r{\'e}nyi entropies: A new generalization and some properties},
  \href{https://doi.org/10.1063/1.4838856}{\emph{Journal of Mathematical
  Physics} {\bfseries 54} (2013) 122203}.

\bibitem{wilde2014strong}
M.~M. Wilde, A.~Winter and D.~Yang, \emph{Strong converse for the classical
  capacity of entanglement-breaking and hadamard channels via a sandwiched
  r{\'e}nyi relative entropy},
  \href{https://doi.org/10.1007/s00220-014-2122-x}{\emph{Communications in
  Mathematical Physics} {\bfseries 331} (2014) 593}.

\bibitem{hiai2017different}
F.~Hiai and M.~Mosonyi, \emph{Different quantum f-divergences and the
  reversibility of quantum operations},
  \href{https://doi.org/10.1142/S0129055X17500234}{\emph{Reviews in
  Mathematical Physics} {\bfseries 29} (2017) 1750023}.

\bibitem{wilde2018optimized}
M.~M. Wilde, \emph{Optimized quantum f-divergences and data processing},
  \href{https://doi.org/10.1088/1751-8121/aad5a1}{\emph{Journal of Physics A:
  Mathematical and Theoretical} {\bfseries 51} (2018) 374002}.

\bibitem{araki1976relative}
H.~Araki, \emph{Relative entropy of states of von neumann algebras},
  \href{https://doi.org/10.2977/prims/1195191148}{\emph{Publications of the
  Research Institute for Mathematical Sciences} {\bfseries 11} (1976) 809}.

\bibitem{petz2003monotonicity}
D.~Petz, \emph{Monotonicity of quantum relative entropy revisited},
  \href{https://doi.org/10.1142/S0129055X03001576}{\emph{Reviews in
  Mathematical Physics} {\bfseries 15} (2003) 79}.

\bibitem{nielsen2005simple}
M.~Nielsen and D.~Petz, \emph{A simple proof of the strong subadditivity
  inequality}, \href{https://doi.org/10.5555/2011670.2011678}{\emph{Quantum
  Information \& Computation} {\bfseries 5} (2005) 507}.

\bibitem{witten2018aps}
E.~Witten, \emph{Aps medal for exceptional achievement in research: Invited
  article on entanglement properties of quantum field theory},
  \href{https://doi.org/10.1103/RevModPhys.90.045003}{\emph{Reviews of Modern
  Physics} {\bfseries 90} (2018) 045003}.

\bibitem{lashkari2019constraining}
N.~Lashkari, \emph{Constraining quantum fields using modular theory},
  \href{https://doi.org/10.1007/JHEP01(2019)059}{\emph{Journal of High Energy
  Physics} {\bfseries 2019} (2019) 59}.

\bibitem{faulkner2020approximate}
T.~Faulkner, S.~Hollands, B.~Swingle and Y.~Wang, \emph{Approximate recovery
  and relative entropy {I}. general von neumann subalgebras},
  {\emph{\href{https://arxiv.org/abs/2006.08002}{arXiv:2006.08002}} (2020) }.

\bibitem{moosa2020renyi}
M.~Moosa, P.~Rath and V.~P. Su, \emph{A renyi quantum null energy condition:
  Proof for free field theories},
  {\emph{\href{https://arxiv.org/abs/2007.15025}{arXiv:2007.15025}} (2020) }.

\bibitem{hollands2020variational}
S.~Hollands, \emph{Variational approach to sandwiched renyi entropies (with
  application to {QFT})},
  {\emph{\href{https://arxiv.org/abs/2009.05024}{arXiv:2009.05024}} (2020) }.

\bibitem{leditzky2017data}
F.~Leditzky, C.~Rouz{\'e} and N.~Datta, \emph{Data processing for the
  sandwiched r{\'e}nyi divergence: a condition for equality},
  \href{https://doi.org/10.1007/s11005-016-0896-9}{\emph{Letters in
  Mathematical Physics} {\bfseries 107} (2017) 61}.

\bibitem{jenvcova2017preservation}
A.~Jen{\v{c}}ov{\'a}, \emph{Preservation of a quantum r{\'e}nyi relative
  entropy implies existence of a recovery map},
  \href{https://doi.org/10.1088/1751-8121/aa5661}{\emph{Journal of Physics A:
  Mathematical and Theoretical} {\bfseries 50} (2017) 085303}.

\bibitem{jenvcova2018renyi}
A.~Jen{\v{c}}ov{\'a}, \emph{R{\'e}nyi relative entropies and noncommutative $
  {L}_p $-spaces},
  \href{https://doi.org/10.1007/s00023-018-0683-5}{\emph{Annales Henri
  Poincar{\'e}} {\bfseries 19} (2018) 2513}.

\bibitem{jenvcova2017r}
A.~Jen{\v{c}}ov{\'a}, \emph{R\'enyi relative entropies and noncommutative $
  {L}_p $-spaces {II}},
  {\emph{\href{https://arxiv.org/abs/1707.00047}{arXiv:1707.00047}} (2017) }.

\bibitem{frank2013monotonicity}
R.~L. Frank and E.~H. Lieb, \emph{Monotonicity of a relative r{\'e}nyi
  entropy}, \href{https://doi.org/10.1063/1.4838835}{\emph{Journal of
  Mathematical Physics} {\bfseries 54} (2013) 122201}.

\bibitem{petz1988sufficiency}
D.~Petz, \emph{Sufficiency of channels over von neumann algebras},
  \href{https://doi.org/10.1093/qmath/39.1.97}{\emph{The Quarterly Journal of
  Mathematics} {\bfseries 39} (1988) 97}.

\bibitem{hayden2004structure}
P.~Hayden, R.~Jozsa, D.~Petz and A.~Winter, \emph{Structure of states which
  satisfy strong subadditivity of quantum entropy with equality},
  \href{https://doi.org/10.1007/s00220-004-1049-z}{\emph{Communications in
  mathematical physics} {\bfseries 246} (2004) 359}.

\bibitem{jenvcova2006sufficiency}
A.~Jen{\v{c}}ov{\'a} and D.~Petz, \emph{Sufficiency in quantum statistical
  inference},
  \href{https://doi.org/10.1007/s00220-005-1510-7}{\emph{Communications in
  mathematical physics} {\bfseries 263} (2006) 259}.

\bibitem{fawzi2015quantum}
O.~Fawzi and R.~Renner, \emph{Quantum conditional mutual information and
  approximate markov chains},
  \href{https://doi.org/10.1007/s00220-015-2466-x}{\emph{Communications in
  Mathematical Physics} {\bfseries 340} (2015) 575}.

\bibitem{Sutter2016a}
D.~Sutter, O.~Fawzi and R.~Renner, \emph{Universal recovery map for approximate
  markov chains},
  \href{https://doi.org/10.1098/rspa.2015.0623}{\emph{Proceedings of the Royal
  Society A: Mathematical, Physical and Engineering Sciences} {\bfseries 472}
  (2016) 20150623}.

\bibitem{wilde2015recoverability}
M.~M. Wilde, \emph{Recoverability in quantum information theory},
  \href{https://doi.org/10.1098/rspa.2015.0338}{\emph{Proceedings of the Royal
  Society A: Mathematical, Physical and Engineering Sciences} {\bfseries 471}
  (2015) 20150338}.

\bibitem{Sutter2016b}
D.~Sutter, M.~Tomamichel and A.~W. Harrow, \emph{Strengthened monotonicity of
  relative entropy via pinched petz recovery map},
  \href{https://doi.org/10.1109/tit.2016.2545680}{\emph{{IEEE} Transactions on
  Information Theory} {\bfseries 62} (2016) 2907}.

\bibitem{Sutter2016}
D.~Sutter, M.~Berta and M.~Tomamichel, \emph{Multivariate trace inequalities},
  \href{https://doi.org/10.1007/s00220-016-2778-5}{\emph{Communications in
  Mathematical Physics} {\bfseries 352} (2016) 37}.

\bibitem{junge2018universal}
M.~Junge, R.~Renner, D.~Sutter, M.~M. Wilde and A.~Winter, \emph{Universal
  recovery maps and approximate sufficiency of quantum relative entropy},
  \href{https://doi.org/10.1007/s00023-018-0716-0}{\emph{Annales Henri
  Poincar{\'e}} {\bfseries 19} (2018) 2955}.

\bibitem{carlen2020recovery}
E.~A. Carlen and A.~Vershynina, \emph{Recovery map stability for the data
  processing inequality},
  \href{https://doi.org/10.1088/1751-8121/ab5ab7}{\emph{Journal of Physics A:
  Mathematical and Theoretical} {\bfseries 53} (2020) 035204}.

\bibitem{gao2020recoverability}
L.~Gao and M.~M. Wilde, \emph{Recoverability for optimized quantum $ f
  $-divergences},
  {\emph{\href{http://arxiv.org/abs/2008.01668}{arXiv:2008.01668}} (2020) }.

\bibitem{araki1982positive}
H.~Araki and T.~Masuda, \emph{Positive cones and lp-spaces for von neumann
  algebras}, \href{https://doi.org/10.2977/prims/1195183577}{\emph{Publications
  of the Research Institute for Mathematical Sciences} {\bfseries 18} (1982)
  759}.

\bibitem{carlen2010trace}
E.~Carlen, \emph{Trace inequalities and quantum entropy: an introductory
  course}, \href{https://doi.org/10.1090/conm/529}{\emph{Entropy and the
  quantum} {\bfseries 529} (2010) 73}.

\bibitem{berta2018renyi}
M.~Berta, V.~B. Scholz and M.~Tomamichel, \emph{R{\'e}nyi divergences as
  weighted non-commutative vector-valued ${L}_p$-spaces},
  \href{https://doi.org/10.1007/s00023-018-0670-x}{\emph{Annales Henri
  Poincar{\'e}} {\bfseries 19} (2018) 1843}.

\bibitem{bhatia2013matrix}
R.~Bhatia, \emph{Matrix analysis}, vol.~169. Springer Science \& Business
  Media, 2013,
  \href{https://doi.org/10.1007/978-1-4612-0653-8}{10.1007/978-1-4612-0653-8}.

\bibitem{bhatia2009positive}
R.~Bhatia, \emph{Positive definite matrices}, vol.~24. Princeton university
  press, 2009.

\bibitem{chehade2020saturating}
S.~Chehade and A.~Vershynina, \emph{Saturating the data processing inequality
  for $\alpha -z$ r{\'e}nyi relative entropy},
  {\emph{\href{https://arxiv.org/abs/2006.07726}{arXiv:2006.07726}} (2020) }.

\bibitem{zhang2020equality}
H.~Zhang, \emph{Equality conditions of data processing inequality for $\alpha
  -z$ r{\'e}nyi relative entropies},
  {\emph{\href{https://arxiv.org/abs/2007.06644}{arXiv:2007.06644}} (2020) }.

\bibitem{Petz1986equality}
D.~Petz, \emph{On the equality in jensen's inequality for operator convex
  functions}, \href{https://doi.org/10.1007/BF01195811}{\emph{Integral
  equations and operator theory} {\bfseries 9} (1986) 744}.

\end{thebibliography}\endgroup

\appendix
\section{Deriving the Petz map}\label{app:jensenequality}
The only inequality that enters the proof of the DPI for the quantum relative entropy is Jensen's operator inequality. 
Thus, the case of equality in the DPI translates into the case of equality in Jensen's operator inequality. 
This case was, perhaps unsurprisingly, studied by Petz in Ref.~\cite{Petz1986equality}, resulting in the following Lemma (which we state in simplified form), which we will use to give a simple, direct derivation of the Petz recovery map.
\begin{lem}
Let $\Phi$ be a normalized positive map between two $C^*$ algebras containing units and let $f:\mathbb R\rightarrow \mathbb R$ be a convex map. 
Then
	\begin{align}
	f(\Phi(a))  \leq \Phi(f(a)),\quad \forall a = a^*.
	\end{align}
	Moreover, if $f$ is non-affin and equality holds for some $a$, then 
	\begin{align}
	\Phi(a^2) = \Phi(a)^2. 
	\end{align}
\end{lem}
We will apply the Lemma to the postive map $\Phi(a) = U^* a U$ with $a=\Delta_{AB} := \Delta^{AB}_{\sigma,\rho}$. 
Then the case of equality in the DPI corresponds to $(U^*a U)^2 = U^* a^2 U$.
\begin{lem}
	Let $UU^* = P$, $U^*U=\mathbf 1$ and $a=a^*$. Then $(U^*a U)^2 = U^* a^2 U$ if and only if $[P,a]=0$ (i.e., $a=PaP+QaQ$ with $Q=\id - P$). 
	\begin{proof}
		We first prove that $[P,a]=0$. To do this, we use $PU=UU^*U = U$. First,
		\begin{align}
			(U^* a U)^2 &= U^* a UU^* aU = U^* aPaU = U^* PaPaP U.
		\end{align}
		Second,
		\begin{align}
			U^* a^2 U&= U^* a(P+Q)aU = U^* aPa U + U^* a Q aU.
		\end{align}
		Conjugating both equations with $U(\cdot)U^*$ and using $(U^*a U)^2 = U^* a^2 U$, we then find
		\begin{align}
		PaPaP = PaPaP + PaQaP = PaPaP + (PaQ)(PaQ)^*.
		\end{align}
		Thus $PaQ = 0 = QaP$, since $bb^*=0$ implies $b^*=0=b$ ($\|b^*\ket{\Psi}\|=0$ for any $\ket{\Psi}$), and therefore $a=PaP+QaQ$.
		
		For the converse direction, we use $a=PaP+QaQ$ and $PU=U$ to find
		\begin{align}
			U^* a^2 U = U^* (PaPaP + QaQaQ) U = U^* PaPaP U = U^* a U U^* a U = (U^*a U)^2.
		\end{align}
	\end{proof}
\end{lem}
Returning to our problem at hand, we find $[\Delta_{AB},P]=0$ and consequently $[\Delta_{AB}^t,P]=0$ for any $t\geq 0$. In particular, we get 
\begin{align} \label{eq:commuting}
 P\Delta^t_{AB}P = (P\Delta_{AB} P)^t.
\end{align}
The operator $\Delta_{AB}^t$ acts as $\Delta_{AB}^t(b)=\sigma_{AB}^tb\rho_{AB}^{-t}$ and similarly for $\Delta_A$.
Therefore $\Delta_{AB}^t(\rho_{AB}^t) = \sigma_{AB}^t$.
Furthermore, $P(\rho^\frac12_{AB}) = \rho^\frac12_{AB}$.
We now specialize to the case $t=1/2$ and evaluate both sides of \eqref{eq:commuting} in different ways. First, for the LHS we get
\begin{equation}	
	P\Delta^\frac12_{AB}P(\rho^\frac12_{AB})= \Delta^\frac12_{AB}P(\rho^\frac12_{AB})=\Delta^\frac12_{AB}(\rho^\frac12_{AB}) = \sigma^\frac12_{AB}.\label{eq:sigmaAB}
\end{equation}
Considering the RHS of \eqref{eq:commuting}, on the other hand, we find using $P=U U^*$ and $U^*\Delta_{AB}U=\Delta_A$: 
\begin{align}
	P\Delta^\frac12_{AB}P(\rho^\frac12_{AB})&=(UU^*\Delta_{AB}UU^*)^\frac12(\rho^\frac12_{AB})=U(U^*\Delta_{AB}U)^\frac12U^*(\rho^\frac12_{AB})\\
	&=U\Delta^\frac12_{A}U^*(\rho^\frac12_{AB})=U\Delta^\frac12_{A}(\rho^\frac12_{A})=U(\sigma_A^\frac12)=(\sigma_A^\frac12\rho_A^{-\frac12}\otimes\id_B)\rho_{AB}^\frac12.
\end{align}
Together with the RHS of \eqref{eq:sigmaAB} we obtain $(\sigma_A^\frac12\rho_A^{-\frac12}\otimes\id_B)\rho_{AB}^\frac12=\sigma_{AB}^\frac12.$ 
Taking the square, we conclude: 
\begin{align}
	\sigma_{AB} &= \rho_{AB}^\frac12(\rho_A^{-\frac12}\sigma_A \rho_A^{-\frac12}\otimes \id_B)\rho_{AB}^{\frac12}= \R_{\rho,\tr_B}(\sigma_A)
\end{align}
and clearly $\R_{\rho,\tr_B}(\rho_A)=\rho_{AB}$. We have thus re-derived the Petz recovery map $\R$. 
Note that all that we used to derive the recovery map is that $[\Delta_{AB},P]=0$, which we followed from the case of equality in Jensen's operator inequality.

\section{Other necessary conditions}
Let's explore some other choices of $\beta$ in Theorem~
\ref{theorem:newequality'} to obtain more necessary conditions. Consider $\beta=\alpha-1,$
\begin{align}\label{eq:equalitycondition5}
 \sigma_{A}^{\alpha-1} \rho_{A}^\frac12 (\rho_A^\frac12\sigma_A^{-\alpha}\rho_A^\frac12)\rho_A^{-\frac12}\otimes \id_B = \sigma^{\alpha-1}_{AB}\rho^\frac12_{AB} (\rho_{AB}^\frac12\sigma_{AB}^{-\alpha}\rho_{AB}^\frac12)\rho_{AB}^{-\frac12}\,,
 \end{align}
 which reduces to
\begin{align}\label{eq:necessary1}
 \rho_{AB} = \sigma^{1-\alpha}_{AB}\sigma_{A}^{\alpha-1} \rho_A\sigma_A^{-\alpha}\sigma_{AB}^{\alpha}\,.
\end{align}
It is in the form of perfect recovery of $\rho_A.$ Define the recovery map $\R_{\alpha,\sigma,\tr_B}$ as
\begin{equation}\label{eq:recoverymap}
\R_{\alpha,\sigma,\tr_B}(\cdot) := \sigma^{1-\alpha}_{AB}\sigma_{A}^{\alpha-1} \cdot \sigma_A^{-\alpha}\sigma_{AB}^{\alpha}\,.
\end{equation}
This recovery map parameterized by $\alpha$ is a one-parameter family generalizations of the usual Petz map, which corresponds to $\alpha=\frac12.$ It's trace-preserving but perhaps it is not positive for all $\sigma$ when $\alpha\neq\frac12$. Therefore, (\ref{eq:equalitycondition5}) alone may not be a sufficient condition for the DPI equality.

Furthermore, setting $\beta=1-\alpha$ yields a rather simple necessary condition for the equality in DPI.
\begin{align}
 \sigma_{A}^{1-\alpha} \rho_{A}^\frac12 (\rho_A^\frac12\sigma_A^{-\alpha}\rho_A^\frac12)^{-1}\rho_A^{-\frac12}\otimes \id_B = \sigma^{1-\alpha}_{AB}\rho^\frac12_{AB} (\rho_{AB}^\frac12\sigma_{AB}^{-\alpha}\rho_{AB}^\frac12)^{-1}\rho_{AB}^{-\frac12}\,,
\end{align}
 which reduces to
 \begin{align}\label{eq:necessary2}
 \sigma_{A}\rho_A^{-1}\otimes \id_B = \sigma_{AB}\rho_{AB}^{-1}\,.
 \end{align}
Again, (\ref{eq:necessary2}) alone is perhaps not sufficient for the DPI equality.

\end{document}